\documentclass[conference]{ieeeconf}
\IEEEoverridecommandlockouts
\usepackage{cite}

\usepackage{amsmath,amssymb,amsfonts}
\usepackage{subfigure}
\usepackage{algorithmic}
\usepackage{graphicx}
\usepackage{textcomp}
\usepackage{xcolor}

\newtheorem{proposition}{Proposition}

\newtheorem{remark}{Remark}
\newtheorem{assumption}{Assumption}

\def\BibTeX{{\rm B\kern-.05em{\sc i\kern-.025em b}\kern-.08em
    T\kern-.1667em\lower.7ex\hbox{E}\kern-.125emX}}

\IEEEoverridecommandlockouts 

\title{\bf A Neural Network Approach towards Real-Time Knot Point Density Estimates for the Nonlinear Quadrotor}

\author{Jasmine Ratchford, Raffaele Romagnoli, and Mark H. Klein
        \thanks{J. Ratchford conducted this work while with the Software Engineering Institute, Carnegie Mellon University, Pittsburgh, PA 15213, USA:
        {\tt\small  jasmine.star.ma@gmail.com}
        }
        \thanks{R. Romagnoli is with the School of Science and Engineering, Duquesne University, Pittsburgh, PA 15213, USA:
        {\tt\small romagnolir@duq.edu}
        }
        \thanks{M. H. Klein is with the Software Engineering Institute, Carnegie Mellon University, Pittsburgh, PA 15235, USA:
        {\tt\small  mk@sei.cmu.edu }
        }
}

\usepackage{fancyhdr}

\fancypagestyle{firstpage}{%
  \fancyhf{}
  \fancyhead[C]{\small[Distribution Statement A] Approved for public release and unlimited distribution}
  \fancyfoot[C]{\small[Distribution Statement A] Approved for public release and unlimited distribution}
  \setlength{\headheight}{0.47in}
  \setlength{\headsep}{-20pt}
 }
\begin{document}
\maketitle
\thispagestyle{firstpage}
\begin{abstract}
Spline representations provide a compact way to approximate continuous trajectories using a finite number of parameters. When approximating trajectories generated by dynamical systems, the efficiency of a spline representation depends strongly on the placement of knot points. Uniform knot distributions often lead to inefficient representations, requiring either unnecessary control points or producing poor approximations in regions where the trajectory varies rapidly.

In this work we propose a learning-based method for adaptive knot placement in B-spline approximations of trajectories generated by ordinary differential equations. A neural network learns a mapping from the initial condition of the system to quantities that characterize the local variation of the resulting trajectory. 
These estimates are used to construct nonuniform knot distributions that approximately equidistribute a local error indicator, reducing the number of control points needed to meet a prescribed accuracy.\footnote{This work has been submitted to the IEEE for possible publication. Copyright may be transferred without notice, after which this version may no longer be accessible.}
\end{abstract}

\section{Introduction}

Modern autonomous systems such as unmanned aerial vehicles, robotic manipulators, and autonomous ground vehicles must predict the future evolution of their dynamics under strict computational constraints \cite{hsu2023safety,siciliano2009robotics, gasparetto2012trajectory}.  Fast and compact representations of system trajectories are required both for motion planning and for safety analysis \cite{wabersich2023data, wursching2021sampling, li2026optimal}. 
Spline-based representations are widely used for this purpose because they provide smooth approximations described by a finite number of  control points \cite{unser1999splines}. Splines are piecewise polynomial \cite{schumaker2007spline} defined over a sequence of knot points that partition the domain into intervals. On each interval the function is described by a polynomial segment, and smoothness conditions are imposed at the knots to ensure continuity of the function and its derivatives across adjacent segments. 

{
B-splines provide a particularly convenient formulation \cite{unser2002b}. The function is expressed as a linear combination of basis functions defined with respect to a knot vector, which specifies the sequence of knot locations and determines the support of each basis function. The coefficients associated with the basis functions are called control points, and their sequence forms the control polygon. This geometric structure provides a compact description of the spline and allows useful bounds on the approximation error to be derived from the relationship between the spline and its control polygon \cite{de1978practical}. B-splines can also be directly employed as a parametric representation of the solution, enabling the approximation of differential equations within the spline space \cite{hughes2005isogeometric, jator2007high,romagnoli2024building}.
}
{
Despite these advantages, constructing efficient B-spline representations remains challenging \cite{rathod2026splineflow, van2015b}. Often, trajectories are obtained by first solving the governing ordinary differential equation and only subsequently approximated using a B-spline.  The efficiency of the representation depends strongly on the placement of knot points \cite{van2015b,rousseau2019minimum, al2025point}. Uniformly spaced knots are commonly used in practice, but they often lead to inefficient representations: simple portions of a trajectory may require only a few control points, while regions where the trajectory varies rapidly may require a much higher density of knots.
} 

Classical results in spline approximation theory relate the error between a spline and its control polygon to the spacing between knot points and to the curvature of the underlying function \cite{de1978practical, prautzsch2002bezier}. {These results motivate adapting the knot distribution to the local variation of the trajectory.} In many dynamical systems the trajectory is uniquely determined by the initial condition together with the governing differential equation \cite{Khalil:1173048}. This observation suggests that information about the complexity of the resulting trajectory—and therefore about an appropriate knot distribution—may be inferred directly from the initial condition.

In this work we explore this idea by learning a mapping from the initial condition of the system to quantities that characterize the local variation of the resulting trajectory. These estimates are then used to determine knot placement for B-spline representations that approximate the solution of the differential equation with a prescribed accuracy. Importantly, the desired approximation tolerance can be specified at inference time, allowing the accuracy of the representation to be adjusted without retraining the model. We first review the spline approximation results that motivate adaptive knot placement and then present the proposed learning framework. We demonstrate that it produces efficient B-spline representations of dynamical trajectories using fewer control points than uniform knot placement while maintaining accuracy.

\section{Preliminaries}\label{sec:preliminaries}

\subsection{Nonlinear Autonomous System}

Consider the autonomous nonlinear system
\begin{equation}\label{eq:closed_loop}
\dot{x} = f(x), \qquad x(0)=x_0 \in D \subset \mathbb{R}^n,
\end{equation}
where $D$ is a compact set containing the origin. 
We assume $f:\mathbb{R}^n \to \mathbb{R}^n$ is locally Lipschitz on $D$. 
For every $x_0 \in D$, the solution $x(\cdot;x_0) : [a,b] \to \mathbb{R}^n$
exists, is unique, and satisfies $x(\cdot;x_0) \in C^2([a,b];\mathbb{R}^n)$ \cite{Khalil:1173048}. 
\subsection{Norms}

Throughout the paper, $\mathbb{R}^n$ is equipped with the Euclidean norm
\[
\|x\| := \sqrt{x_1^2 + \dots + x_n^2}.
\]

For a continuous trajectory $x : [a,b] \to \mathbb{R}^n$, 
the $L^p$ norm ($1 \le p < \infty$) is defined as
\[
\|x\|_{L^p([a,b])}
=
\left(
\int_a^b \|x(t)\|^p \, dt
\right)^{1/p},
\]
and the $L^\infty$ norm as
\[
\|x\|_{\infty}
=
\sup_{t\in[a,b]} \|x(t)\|.
\]

Since $x$ is continuous on the compact interval $[a,b]$, 
the supremum is attained and $\|x\|_{\infty}<\infty$ \cite{sawano2018theory}.

\subsection{B-Splines}

Let $d \ge 1$ denote the spline degree and 
$\hat{\mathbf t} = (\hat t_1,\dots,\hat t_{\ell+d+1})$ 
a nondecreasing knot vector on $[a,b]$ with clamped endpoints:
\[
\hat t_1 = \cdots = \hat t_{d+1} = a,
\qquad
\hat t_{\ell+1} = \cdots = \hat t_{\ell+d+1} = b.
\]

The B-spline basis functions $B_{j,d}(t)$ are defined recursively 
via the Cox--de Boor formula \cite{de1978practical, schumaker2007spline}. 
They satisfy
\[
B_{j,d}(t) \ge 0,
\qquad
\sum_{j=1}^{\ell} B_{j,d}(t) = 1,
\qquad \forall t\in[a,b].
\]

A vector-valued spline $s : [a,b] \to \mathbb{R}^n$ 
is given by
\[
s(t) = \sum_{j=1}^{\ell} c_j B_{j,d}(t),
\qquad c_j \in \mathbb{R}^n.
\]

Where $\{c_i\}$ are the set of control points. By non-negativity and partition of unity,
\[
s(t) \in \operatorname{conv}\{c_1,\dots,c_\ell\},
\qquad \forall t\in[a,b].
\]
where $conv\{\}$ indicates the convex hull.
\begin{remark}
In this paper we consider a UAV application to avoid ambiguity with the state vector $x(t)$, we denote the position components by $p = (p_x, p_y, p_z)$ and the corresponding velocities by $\dot{p} = (v_x, v_y, v_z)$.  
\end{remark}

\section{Problem Formulation}

We consider a controlled nonlinear system of the form
\begin{equation}\label{eq:controlled_system}
\dot x = f(x,u),
\end{equation}
equipped with a smooth state-feedback controller $u = -\gamma(x),$ so that the closed-loop dynamics become autonomous:
\begin{equation}\label{eq:closed_loop_pf}
\dot x = f(x).
\end{equation}

We assume that the closed-loop system satisfies the regularity
properties introduced in Section~\ref{sec:preliminaries}. 
In particular, for every initial condition $x_0 \in D$,
there exists a unique trajectory $x(\cdot;x_0) \in C^2([a,b];\mathbb{R}^n).$ In practice, the trajectory is available only through sampled evaluations 
$x(t_k; x_0)$ obtained via numerical integration of the differential equation. These samples are used to construct the spline approximation. {Our objective is to approximate $x(\cdot;x_0)$ by a spline of fixed degree $d$ and $\ell$ control points, while adaptively selecting the internal knot partition as a function of $x_0$.}

\subsection{Spline Representation and Least-Squares Control Points}

Let $\hat{\mathbf t}$ denote a clamped knot vector on $[a,b]$
with internal knots $\boldsymbol\tau$.
Define the basis vector
\[
\mathbf b(t)
=
\begin{bmatrix}
B_{1,d}(t) & \cdots & B_{\ell,d}(t)
\end{bmatrix}^{\!\top}.
\]

Let $C \in \mathbb{R}^{\ell \times n}$ collect the control points,
so that the spline approximation can be written compactly as
\begin{equation}\label{eq:spline_matrix_form}
s(t;\boldsymbol\tau,C)
=
C^\top \mathbf b(t).
\end{equation}

Given sampling instants $a \le t_1 < \dots < t_N \le b$,
define the design matrix
\[
\Phi(\boldsymbol\tau)
=
\begin{bmatrix}
\mathbf b(t_1)^\top\\
\vdots\\
\mathbf b(t_N)^\top
\end{bmatrix}
\in \mathbb{R}^{N \times \ell},
\]
and the sampled trajectory matrix
\[
X(x_0)
=
\begin{bmatrix}
x(t_1;x_0)^\top\\
\vdots\\
x(t_N;x_0)^\top
\end{bmatrix}
\in \mathbb{R}^{N \times n}.
\]

For a fixed partition $\boldsymbol\tau$,
the control points are obtained via least squares:
\begin{equation}\label{eq:ls_problem_final}
C^\star(\boldsymbol\tau;x_0)
\in
\arg\min_{C \in \mathbb{R}^{\ell \times n}}
\|\Phi(\boldsymbol\tau) C - X(x_0)\|_F^2.
\end{equation}

\subsection{{From Optimal to Approximate Knot Placement}}

For a given initial condition $x_0$,
the optimal internal knot partition is defined as
\begin{equation}\label{eq:knot_opt_L2}
\boldsymbol\tau^\star(x_0)
\in
\arg\min_{\boldsymbol\tau}
\;
\|x(\cdot;x_0)
-
s(\cdot;\boldsymbol\tau,
C^\star(\boldsymbol\tau;x_0))
\|_{L^2([a,b])}.
\end{equation}

{Computing \eqref{eq:knot_opt_L2} exactly is generally intractable. We emphasize that the optimal control points $C^\star(\boldsymbol\tau;x_0)$ depend on the knot partition $\boldsymbol\tau$, making the free-knot problem nonlinear and generally nonconvex. The proposed method does not solve \eqref{eq:knot_opt_L2} directly. Instead, we exploit classical approximation theory to guide partition selection. In particular, Jackson-type inequalities for spline approximation provide local error estimates involving
\[
h_i^2 \|x''\|_{L^\infty([t_i,t_{i+1}])},
\]
where $h_i$ denotes the local knot spacing. This motivates the curvature-based local error indicator and the associated design objective}
\begin{equation}\label{eq:surrogate_obj_final}
\min_{\boldsymbol\tau}
\;
\max_i
\left(
h_i^2 \|x''\|_{L^\infty([t_i,t_{i+1}])}
\right)
\end{equation}
which favors smaller intervals in regions of higher curvature. For notational simplicity, when the initial condition $x_0$ is fixed, we denote the corresponding trajectory by $x(\cdot) := x(\cdot;x_0)$.

\section{Theoretical Results}

This section provides a theoretical justification for the curvature-based knot selection strategy used in the proposed learning framework. The analysis is based on classical approximation theory and Jackson-type inequalities for spline approximation.

\subsection{Linear Interpolation and Spline Surrogates}

We begin with a classical result for linear interpolation.

\begin{proposition}[Linear Interpolation \cite{de1978practical}]\label{prop:linear_bound}
Let $x \in C^2([a,b];\mathbb{R}^n)$ and let
$a=t_0<t_1<\dots<t_m=b$ be a partition with
$h_i=t_{i+1}-t_i$.
Let $\mathcal I_1 x$ denote the component-wise
piecewise-linear interpolant of $x$ on this partition.
Then for each sub-interval $[t_i,t_{i+1}]$,
\begin{equation}\label{eqn:polygon_error}
\|x-\mathcal I_1 x\|_{L^\infty([t_i,t_{i+1}])}
\le
\frac{h_i^2}{8}
\|x''\|_{L^\infty([t_i,t_{i+1}])}.
\end{equation}
\end{proposition}
\vspace{0.1in}
\noindent
Proposition~\ref{prop:linear_bound} shows that the local interpolation
error scales quadratically with the interval length and linearly with
the curvature magnitude. In particular, the quantity
\begin{equation}\label{eq:spline_accuracy}
h_i^2\,\|x''\|_{L^\infty([t_i,t_{i+1}])}    
\end{equation}

naturally measures the local approximation difficulty, suggesting that
smaller intervals should be allocated in regions where the curvature
$\|x''\|$ is large.





\subsection{Jackson-Type Spline Approximation Bound}

A classical Jackson inequality for spline approximation states that, for
$x(\cdot)\in C^r([a,b];\mathbb{R}^n)$ with $1\le r\le d+1$, the best
approximation error by degree-$d$ splines with $m$ knot spans satisfies
\begin{equation}\label{eqn:jackson}
\begin{aligned}
\inf_{s\in\mathcal S_{d,\boldsymbol\tau_m}}
\|x(\cdot)-s\|_{L^\infty([a,b])}
\le\;& (b-a)^r\, m^{-r} \\
&\times \|x^{(r)}(\cdot)\|_{L^\infty([a,b])}\, C_{d,r},
\end{aligned}
\end{equation}
where $C_{d,r}>0$ depends only on $(d,r)$ and is independent of $m$ and $x$\cite{devore1998nonlinear}.
The inequality implies the convergence rate
\[
\inf_{s}\|x-s\|_{L^\infty}=O(m^{-r}).
\]

Since the mesh size satisfies $h \approx (b-a)/m$,
inequality \eqref{eqn:jackson} can equivalently be written in terms of $h$.

\medskip

In the present setting we assume $x(\cdot)\in C^2([a,b];\mathbb{R}^n)$.
Specializing \eqref{eqn:jackson} to $r=2$ yields 
\begin{equation}\label{eqn:jackson_h}
\inf_{s\in\mathcal S_{d,\boldsymbol\tau}}
\|x-s\|_{L^\infty([a,b])}
\le
C_d\, h^2\,
\|x''\|_{L^\infty([a,b])},
\end{equation}
for some constant $C_d>0$ depending only on the spline degree $d$ \cite{fasshauer_578_notes}.

The Jackson inequality characterizes the \emph{best achievable} approximation error, but it does not prescribe how knots should be distributed in a nonuniform partition. In particular, the bound in \eqref{eqn:jackson_h} depends on the global mesh size $h=\max_i h_i$, which may be overly conservative when the curvature of $x$ varies across the domain. To obtain insight into the distribution of approximation error, it is necessary to consider local error contributions. Combining the local scaling identified in \eqref{eq:spline_accuracy}
with the global rate in \eqref{eqn:jackson_h} motivates using the
largest local term
\[
h_i^2\|x''\|_{L^\infty([t_i,t_{i+1}])}
\]
as an indicator of approximation difficulty for nonuniform
partitions.
This observation provides the link between global optimal approximation rates and local mesh adaptation, and motivates the introduction of a curvature-weighted surrogate functional.

\subsection{Curvature-Based Surrogate Functional}\label{sec:curvature}

Motivated by the preceding results, we consider the surrogate optimization problem \eqref{eq:surrogate_obj_final} written as
\begin{equation}\label{eq:surrogate_problem}
\min_{\boldsymbol\tau} \mathcal J(\boldsymbol\tau),
\end{equation}
where
\begin{equation}\label{eq:surrogate_functional}
\mathcal J(\boldsymbol\tau)
:=
\max_{0 \le i \le m-1}
\left(
h_i^2
\|x''\|_{L^\infty([t_i,t_{i+1}])}
\right),
\end{equation}
and $\boldsymbol\tau=\{t_i\}_{i=0}^m$ denotes the partition.

The functional $\mathcal J$ reflects the local $h^2$-scaling predicted by both Jackson-type estimates and classical interpolation bounds. Minimizing $\mathcal J$ therefore promotes a redistribution of the mesh so that the dominant local error contributions are balanced across intervals.

\subsection{Equidistribution Principle}

We now introduce an approximate equidistribution principle for the
surrogate functional under a local frozen-curvature assumption.

{\begin{assumption}
\label{ass:local_curvature}
For sufficiently fine partitions and small boundary perturbations,
$\|x''\|_{L^\infty([t_i,t_{i+1}])}$ is locally constant to first order,
so changes in the local indicators are dominated by changes in $h_i$.
\end{assumption}}

\begin{proposition}[Approximate Equidistribution Principle]
\label{prop:equidistribution}
Let $\boldsymbol\tau=\{t_i\}_{i=0}^m$ be a partition of $[a,b]$
with a fixed number of intervals $m$, and consider the surrogate
optimization problem \eqref{eq:surrogate_problem} under
Assumption~\ref{ass:local_curvature}.

If $\boldsymbol\tau^\star$ is a minimizer of the locally approximated
problem, then the local indicators satisfy
\[
h_i^2
\|x''\|_{L^\infty([t_i,t_{i+1}])}
\approx
\mathcal J(\boldsymbol\tau^\star)
\]
across all active intervals. In particular,
\[
h_i^2
\|x''\|_{L^\infty([t_i,t_{i+1}])}
\approx
\mathrm{const}.
\]

In the continuum limit, writing $h(t)\approx 1/\sigma(t)$, the
corresponding equidistributed knot density satisfies
\[
\sigma(t)
\propto
\sqrt{\|x''(t)\|}.
\]
\end{proposition}

\begin{proof}
Define the local indicators
\[
\eta_i(\boldsymbol\tau)
:=
h_i^2\,\|x''\|_{L^\infty([t_i,t_{i+1}])},
\qquad
h_i=t_{i+1}-t_i,
\]
so that
\[
\mathcal J(\boldsymbol\tau)
=
\max_i\eta_i(\boldsymbol\tau).
\]

In accordance with Assumption~\ref{ass:local_curvature}, the local
curvature quantities are treated as constant to first order under the
small perturbations considered below. Hence, the leading-order changes
in the indicators are determined by changes in the interval lengths.

Let $\boldsymbol\tau^\star$ be a minimizer of the locally approximated
problem among all partitions with $m$ intervals, and define
\[
\lambda:=\mathcal J(\boldsymbol\tau^\star).
\]
Suppose, by contradiction, that the local indicators are not balanced.
Define the sets of active and inactive intervals as
\[
\mathcal A
:=
\{i:\eta_i(\boldsymbol\tau^\star)=\lambda\},
\qquad
\mathcal B
:=
\{j:\eta_j(\boldsymbol\tau^\star)<\lambda\}.
\]
Under the contradiction hypothesis, $\mathcal B$ is nonempty.

Consider a sufficiently small perturbation that reduces the length of
each active interval by $\delta>0$ and redistributes the removed length
among the inactive intervals while preserving
\[
\sum_{i=0}^{m-1}h_i=b-a.
\]
For each $i\in\mathcal A$, let $h_i'=h_i-\delta.$ Since the local curvature quantity is treated as fixed and $h_i'<h_i$,
\[
\eta_i'
:=
(h_i')^2
\|x''\|_{L^\infty([t_i,t_{i+1}])}
<
\eta_i(\boldsymbol\tau^\star)
=
\lambda.
\]

Because every inactive interval initially satisfies
\[
\eta_j(\boldsymbol\tau^\star)<\lambda,
\qquad j\in\mathcal B,
\]
and the indicators depend continuously on the interval lengths under
the local frozen-curvature approximation, $\delta$ can be selected
sufficiently small so that the redistributed length does not cause any
inactive indicator to reach $\lambda$. Therefore,
\[
\eta_j'<\lambda,
\qquad j\in\mathcal B.
\]

It follows that
\[
\mathcal J(\boldsymbol\tau')
=
\max_i\eta_i'
<
\lambda
=
\mathcal J(\boldsymbol\tau^\star),
\]
which contradicts the optimality of $\boldsymbol\tau^\star$ for the
locally approximated problem. Therefore, the local indicators must be
approximately balanced:
\[
h_i^2
\|x''\|_{L^\infty([t_i,t_{i+1}])}
\approx
\mathrm{const}.
\]

In the continuum limit, imposing $h(t)^2\|x''(t)\| \approx \mathrm{const}$ and using $h(t)\approx 1/\sigma(t)$ gives $\sigma(t) \propto \sqrt{\|x''(t)\|}.$

\end{proof}

We emphasize that $\mathcal J$ is a local error indicator rather than
the exact least-squares objective in \eqref{eq:knot_opt_L2}. Its
structure motivates learning curvature-related quantities from the
initial condition to construct an adaptive knot distribution.

\section{Approach}

{The results of Section IV motivate approximately equidistributing the local indicators in \eqref{eq:spline_accuracy} across knot intervals.} This suggests that an effective knot selection strategy requires knowledge of the local variation of the trajectory, as captured by $\|x''\|$.

However, evaluating $\|x''\|$ directly requires access to the full trajectory, which in turn requires solving the underlying differential equation. To overcome this limitation, we approximate these curvature-related quantities using a learned model that maps the initial condition $x_0$ to a surrogate description of the trajectory's local variation.

\subsection{Learned Knot Density}

Given an initial condition $x_0$, our goal is to estimate a time-varying continuous knot density function $\rho(t)$ that reflects the local complexity of the trajectory. Motivated by the scaling laws derived in Section IV, regions of higher curvature should correspond to higher knot density.

In particular, using the local interpolation bound
\[
\|x - I_1 x\| \lesssim \frac{h_i^2}{8} \|x''\|,
\]
we define the knot density (a continuous version of the inverse of $h_i$) as
\[
\rho(t) = \sqrt{\frac{\|x''(t)\|}{8\, \varepsilon}},
\]
where $\varepsilon > 0$ is a prescribed approximation tolerance. {Under the local approximation of Section IV, this choice approximately equidistributes the local error indicator.}

Since $\|x''(t)\|$ is not directly available and requires a choice in metric, we train a neural network to predict a surrogate quantity $\hat{g}(t; x_0)$ that approximates the magnitude of the second derivative along the trajectory under a chosen weighting of the state dimensions. The predicted density is then given by
\[
\rho(t) = \sqrt{\frac{\hat{g}(t; x_0)}{8\, \varepsilon}}.
\]
{The weighting is selected to normalize state components with different physical units and scales.}

\subsection{Knot Placement from Density}

Once $\rho(t)$ is estimated, knot locations are obtained by integrating the density over time. Specifically, define the cumulative density
\[
R(t) = \int_0^t \rho(\tau)\, d\tau.
\]
Knots are placed at time instants where $R(t)$ increases by one unit, which corresponds to selecting knot locations according to the quantiles of the cumulative density. {While Section IV considers a fixed number of intervals, here $R(b)$ determines the number of knot spans for the prescribed tolerance, whereas the quantiles of $R(t)$ determine their relative locations.} To ensure robustness and practical feasibility, we impose bounds on the density:
\[
\rho_{\text{final}}(t) = \max\left(\rho_{\min}, \min(\rho_{\max}, \rho(t))\right),
\]
which limits the minimum and maximum knot spacing.

\subsection{Application to Quadrotor Dynamics}

We now instantiate the proposed framework on a nonlinear quadrotor system \cite{romagnoli2024building}. The state consists of the position components $p = (p_x, p_y, p_z)$, the orientation angles $(\theta,\phi,\psi)$, and their corresponding velocities.

We focus on the velocity states
\[
(v_x, v_y, v_z, \dot{\theta}, \dot{\phi}, \dot{\psi}),
\]
where $(v_x, v_y, v_z) = \dot{p}$. The position variables can be recovered by integration and are not directly used in the control loop. In this setting, the curvature-related quantity $\|x''\|$ corresponds to the second derivative of the velocity components (jerk), which provides a meaningful measure of local trajectory variation. Since the true trajectory $x(t)$ is not available in closed form, we approximate this quantity using the spline representation $s(t)$ constructed from sampled data, and use $\|s''\|$ as a proxy for $\|x''\|$. To combine translational and rotational components, each velocity is normalized by its corresponding physical limits, yielding a dimensionless measure of variation across all components.

The target signal for training is obtained by applying a second-derivative Gaussian convolution (Laplacian of Gaussian) to the velocity trajectories and extracting the maximum magnitude across all components at each time step. This provides an estimate of
\[
\|s''\|_{L^\infty([t_i,t_{i+1}])},
\]
which serves as supervision for the neural network. At inference time, the network predicts the jerk envelope $\hat{g}(t; x_0)$ directly from the initial condition. This prediction is then used to construct the knot density $\rho(t)$ and, through integration, the corresponding knot locations.

\section{Experimental Design}
\subsection{Data}
For each initial condition, trajectories are generated via numerical integration of the nonlinear dynamics and subsequently approximated using B-splines, from which the surrogate curvature signal is computed. {The controller is treated as part of the closed-loop trajectory generator and is not designed or analyzed in this work; it stabilizes the quadrotor at a prescribed hovering setpoint.}
The physical constants we use in the quadrotor solver are given in Table \ref{tab:experiment_info}. 
Input initial conditions are normalized by the R, the radius used to express the ball of initial conditions, also given in \ref{tab:experiment_info}. We calculate the trajectories from initial conditions using a Python numba ODE numerical integration. Because it is possible, we generate new data every epoch, rendering the term epochs somewhat meaningless.  We create 100 batches of 1000 trajectories per epoch. 

For the purposes of calculating the jerk maximum, we normalize translational velocities by 5 meters per second, roll and pitch velocities by  12 radians per second, and yaw velocities by  6 radians per second. This is roughly inspired by drones such as the CrazyFlie\cite{crazyflie} and HolyBro\cite{holybro}, but is not specific to either.  
We do not account for the potential limits in jerk or velocity resolution as there is limited information and these measures tend to be very platform specific. The full list of relevant parameters used is given in Table \ref{tab:experiment_info}.

\begin{table}[]
    \centering
    \begin{tabular}{rl}
        \hline
        \textbf{Parameter}& \textbf{Value}\\
        \hline
        \multicolumn{2}{c}{\textit{Quadrotor Dynamics }}\\
        \hline
        Total system mass              & 0.600 kg\\
        Quadrotor body mass            & 0.410 kg\\
        Propeller radius               & 0.0504 meters\\
        Arm length                     & 0.108 meters\\
        Propeller mass                 & 0.00311 kg\\
        Motor + propeller mass         & 0.0391 kg\vspace{0.1cm}\\
        
        \hline
        \multicolumn{2}{c}{\textit{Initial Conditions Range (maximum magnitudes)}}\\
        \hline
        $x^{ic}_{\max}$& 5.0 meters / second\\
        $\dot{x}^{ic}_{\max}$& 3.0 meters / second\\
        $\theta^{ic}_{\max}$& $\pi/4$ radians\\
        $\dot{\theta}^{ic}_{\max}$& $\pi/2$ radians / second\vspace{0.1cm}\\

        \hline
        \multicolumn{2}{c}{\textit{Discretization and Processing }}\\
        \hline
        Minimum Knot Density	&2 knots per second	\\
        Maximum Knot Density    &40 knots per second\\
        Upper bound on error	&0.001\\
        Translational Rate Max 	&5 meters / second	\\
        Angular Rate Max (roll, pitch)&12 radians / second	\\
        Angular Rate Max (yaw)&6 radians / second\\
        Laplacian of Gaussian $\sigma$&0.01 seconds\\
        $t_{samp}$&0.005 seconds\\
        T (total time) &2.5 seconds\\
    \end{tabular}
    \vspace{0.1cm }
    \caption{Parameter values used in data generation and processing}
    \label{tab:experiment_info}
\end{table}

\subsection{Network Architecture}
We used a network architecture inspired by image architectures with upsamping convolutions to predict $\|s''\|_{L^\infty([t_i,t_{i+1}])}$ 
from initial conditions. The components include feed-forward embedding layers followed by three upsampling convolutions and a final convolution to combine channels. Because real-time processing can impose compute limitations, we attempted to limit the size of the neural network to less than 100,000 parameters. Practically speaking, there is also limited value to high accuracy estimates of $\|s''\|_{L^\infty([t_i,t_{i+1}])}$
, given that the distribution predicted is only being used to estimate distances between knot points. Our final neural network used for the analysis presented in the next section has only 22,395 parameters. Our network architecture used for the analysis that follows is shown in Table \ref{tab:network}. 
\begin{table}
    \centering
        \begin{tabular}{lrcc}\hline
            \textbf{Layer Type} & \textbf{Input} & \textbf{Output}  & \textbf{Param.}\\
            \hline
            Linear / SiLU & 12 & 64 & 832\\
            Linear / SiLU & 64 & 64 & 4160\\
            Linear / SiLU & 64 & 64 & 4160\\
            Linear / SiLU + Reshape +  & 64 & (6,32) & 12492\\
            BatchNorm & & & \\
            Upsample + Conv1d/SiLU& (6,32) & (6,128) & 186\\ 
            Upsample + Conv1d/SiLU& (6,128) & (6,256) & 186\\ 
            Upsample + Conv1d/SiLU& (6,256) & (6,501) & 186\\ 
            Conv1d / SiLU + residual & (6,501) & (6,501) & 186\\ 
            Conv1d (/SiLU) (+ ReLU(at inference)) & (6,501) & (1,501) & 7\\ 

    \end{tabular}\vspace{0.1cm}
    \caption{The architecture used to predict the jerk activity envelope from initial conditions.}
    \label{tab:network}
\end{table}

We use a learning rate of 0.001 combined with an Adam optimizer and a mean squared error loss. We achieve a best validation loss of $0.08$. Training occurs with early stopping. For the run used in the following analysis, it reached the maximum 100 epochs, possibly suggesting that more training could reduce the error further. We posit that further training is unnecessary and additional work to optimize the neural network to perform similarly with fewer parameters or demonstrating its value at half precision is more valuable. 

\section{Results}
  \begin{figure}[h]
      \centering
      \includegraphics[width=0.85\linewidth]{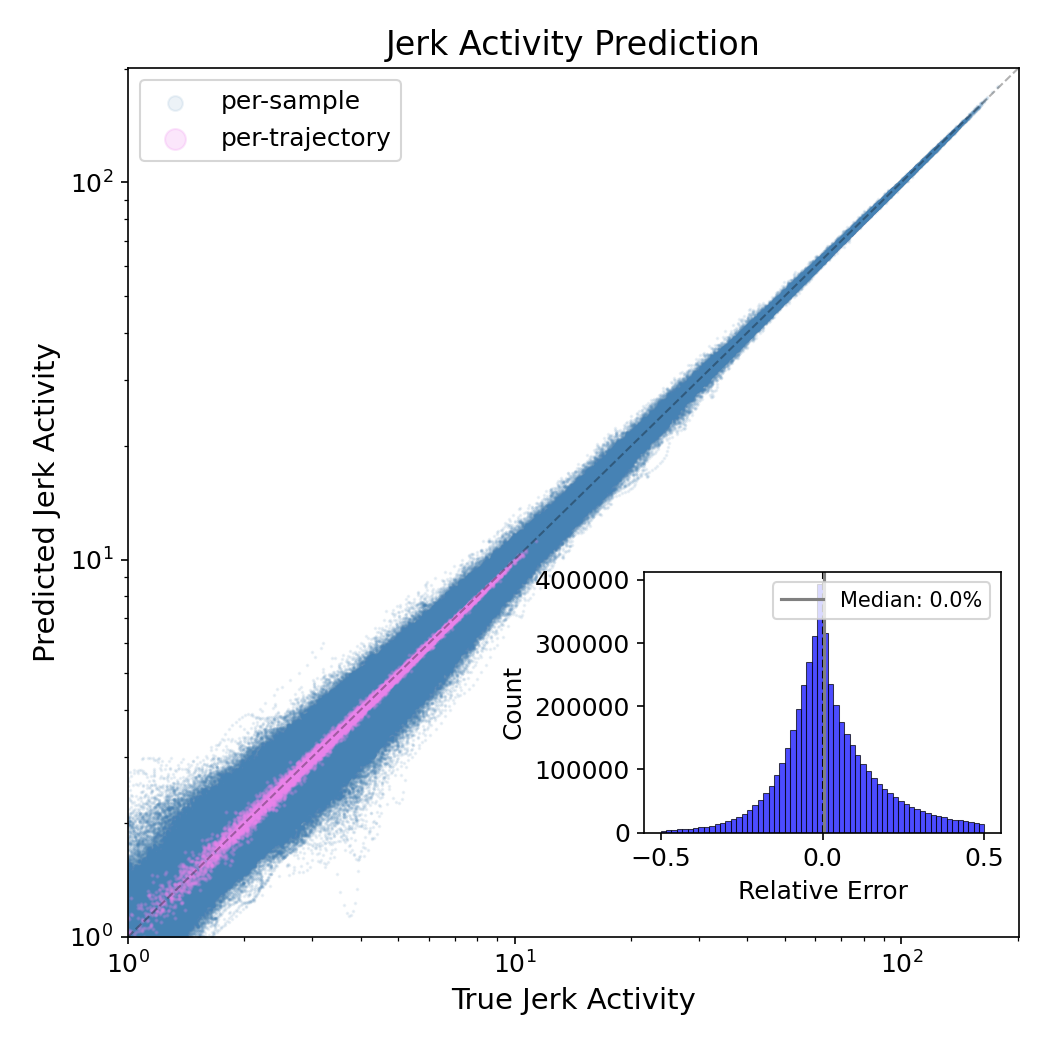}
      \caption{A scatter plot comparing the true jerk activity (x-axis) to the predicted activity (y-axis) (blue) for every sample location along each trajectory, for every trajectory, (purple) averaged per trajectory. The inset histogram shows the distribution of the per-sample jerk percent relative error. The predicted jerk activity matches the true jerk activity well at a per-sample and per-trajectory level.}
      \label{fig:jerk_scatter}
  \end{figure}
Central to our hypothesis is that a small neural network can estimate the envelope of jerk, or 'jerk activity' across the six velocity states. Figure \ref{fig:jerk_scatter} shows a scatter of the true and estimated jerk envelope at a per-sample (e.g. for each sampling point on the trajectory) and at a per-trajectory (e.g. the mean jerk for each trajectory) level. Overall, there is a less than 25\% relative error as shown in the inset histogram. From the larger diagram, its clear that the relative error has a wider distribution at low jerk activity. The mean activity also shows great correspondence, giving more reason to suggest that the overall \emph{number} of knot points predicted is correct.

The distribution of total knot points per trajectory derived using our predicted jerk envelope and the procedure outlined in the Approach section is given in Figure \ref{fig:n_knots_hist}. Again, this meets our expectations, with the number of knots having a broad distribution between 6-8 knots (slightly above our minimum value for number of knots) to about 20 knots. Because quadrotors are most dynamic early in their trajectories, we expect the later parts of the trajectory to be quiescent and lead to a lower knot number.
\begin{figure}
    \centering
    \includegraphics[width=0.7\linewidth]{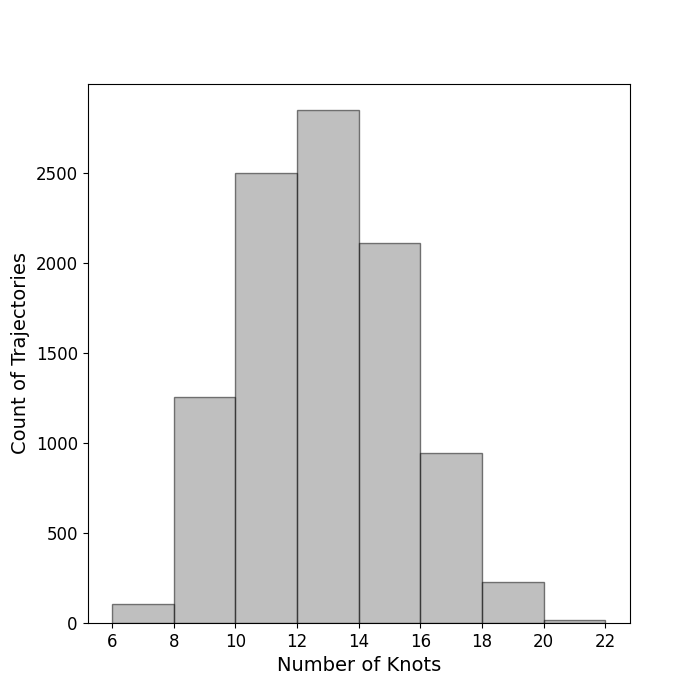}
    \caption{A histogram showing the number of trajectories for which our algorithm prescribed a given number of knots based on a prediction of its jerk throughout the trajectory.}
    \label{fig:n_knots_hist}
\end{figure}

An example of the predicted knot points using the predicted jerk envelope is shown in Figure \ref{fig:example_combined}. This single example shows several typical features, such as a higher density of knot points (given by the ticks at the bottom of the figure) near the start of the trajectory.  We compare this to evenly spaced knot points of the same number. This is shown in orange. The blue and orange scatter points are the Greville abscissae located control points for this trajectory using the predicted (blue) and evenly spaced (orange) knot points. {Both uniformly and adaptively placed knot points describe the trajectory well, with differences visible in the residuals.} However, the control points used with the evenly spaced knots are farther from the trajectory. For evenly  spaced knots as compared to the predicted knot locations, there are more knots in the final, slow dynamical region, and fewer knots in the highly dynamical early trajectory. {For a fair comparison, both methods use the same number of knots, as determined by the prescribed-density procedure.}

\begin{figure}
    \centering
    \begin{subfigure}
        \centering
        \includegraphics[width=0.85\linewidth]{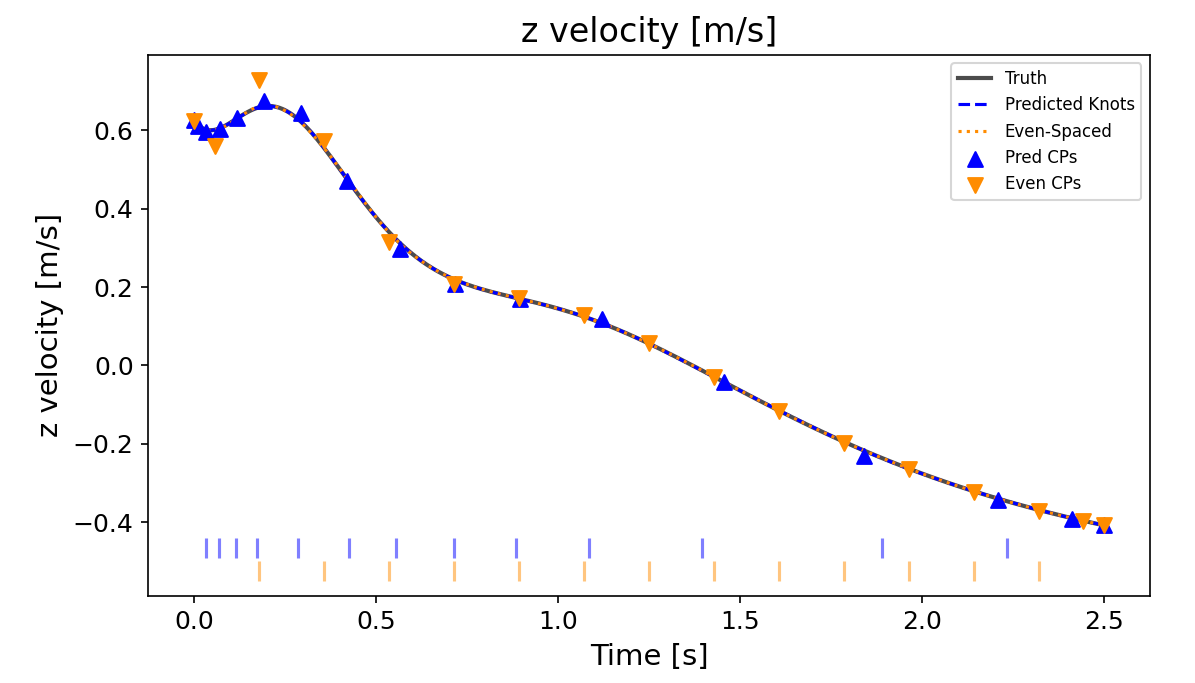}
        \label{fig:example_z}
    \end{subfigure}
    
    \begin{subfigure}
        \centering
        \includegraphics[width=0.85\linewidth]{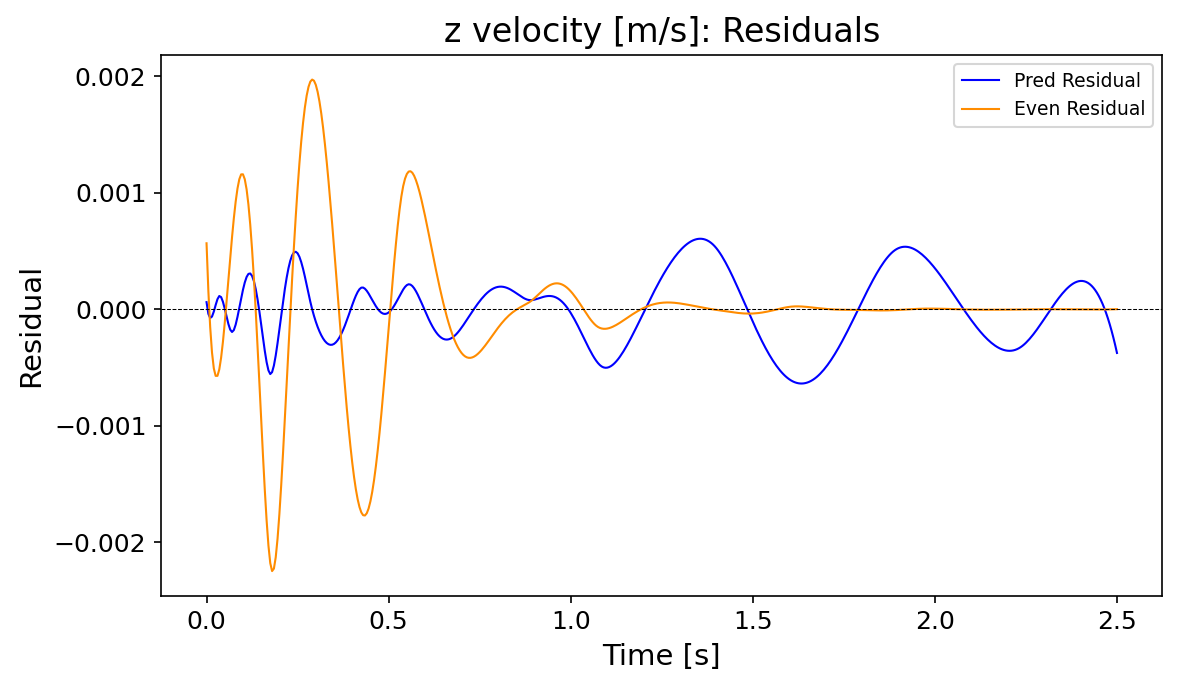}
        \label{fig:example_res}
    \end{subfigure}
    \caption{ Top: A single example trajectory in a single dimension. The true trajectory is shown in black, but is mostly overlaid by the spline fits to the trajectory using the predicted (blue) and evenly spaced (orange) knot points. The knot point locations are given by ticks near the x-axis of the plot. The control point Greville abscissa locations are given by the carrots that match the spline fit color. Bottom: residuals between the true trajectory and each of the spline fits.}
    \label{fig:example_combined}
\end{figure}

A comparison of the control point distance from the trajectory and the associated jerk activity at the nearest sampling point in Figure \ref{fig:cp_v_jerk} shows that the predicted control points maintain a fairly close distance to the trajectory, rising somewhat at low jerk activity. The distribution of knots sacrifices some accuracy at low jerk activity in order to maintain an overall error less than the specified value. In contrast, even control points are associated with a larger distance from the trajectory overall and especially at jerk activity higher than the lowest levels common in the late periods of the trajectory's progress. Shown are only control point dimensions associated with translational velocities. Rotational velocities show a similar pattern.

\begin{figure}
    \centering
    \includegraphics[width=0.85\linewidth, trim={18.1cm 11.5cm 0.5cm 3.5cm}, clip]{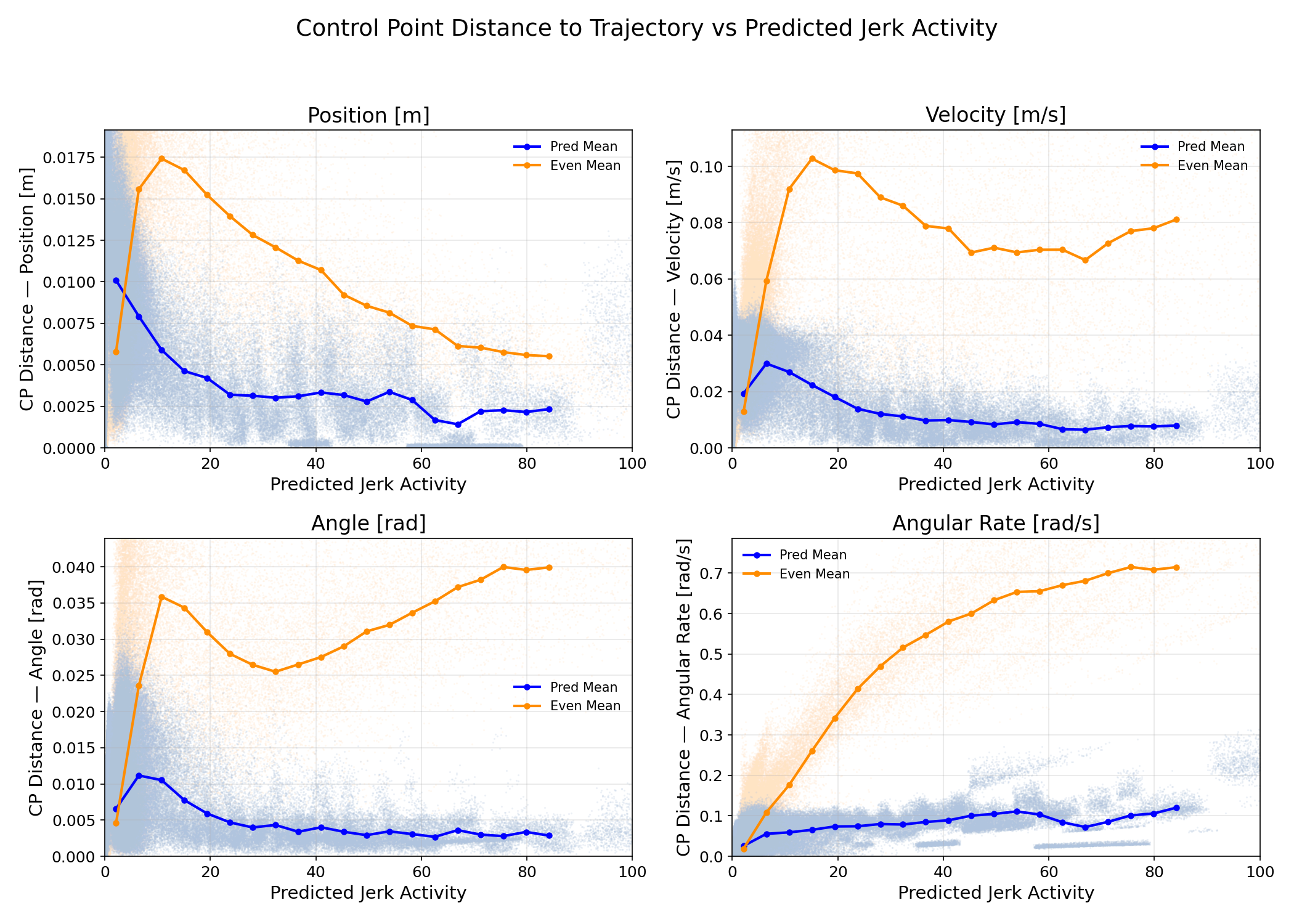}
    \caption{Control point distance from the trajectory compared to the jerk activity associated with the nearest sample point. Blue: control points predicted with the jerk activity method. Orange: control points associated with even distribution. Individual control points are shown as a light, semi-transparent scatter of points. Mean values are shown in opaque colors with connecting lines.}
    \label{fig:cp_v_jerk}
\end{figure}

Figure \ref{fig:rmse_nknots} shows that the RMSE increases as the number of knots decreases for evenly spaced knots, but remains much lower when knot locations are chosen using our prediction method. This is expected. Trajectories with early corrective actions followed by a slow approach to the reference point benefit from an uneven distribution of knots — dense during the initial correction and sparse once the reference is essentially reached. However, trajectories that require more knot points overall tend to have a smaller slow-approach region, meaning their complexity is spread more uniformly across the duration. In such cases, evenly spaced knots more realistically capture the trajectory dynamics.


We evaluated the neural-network inference time on a
Raspberry Pi Zero 2 W with a quad-core 64-bit ARM Cortex-A53 processor and 512~MB of SDRAM. Among 11,000 initial conditions, the first 1,000 were used for warm-up and excluded from the timing results. The mean inference time was 2.07~ms per initial condition, with 95\% of evaluations completed within 2.12~ms. This measurement excludes knot construction and B-spline control-point fitting.

\begin{figure}
    \centering
    \includegraphics[width=0.85\linewidth]{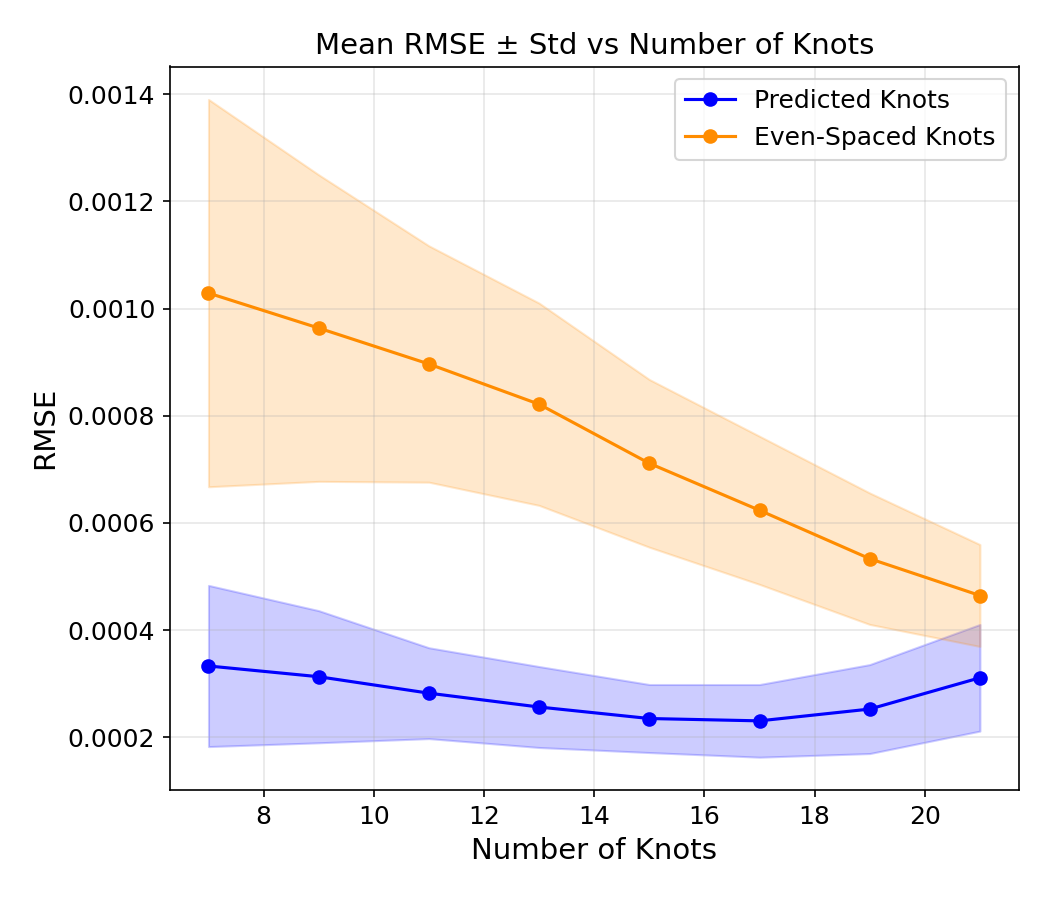}
    \caption{The root mean squared error (RMSE) comparing the spline predicted trajectory with the true trajectory compared the number of knots suggested by the prediction method used for evenly spaced (orange) and jerk-predicted (blue) knot locations. The RMSE for jerk-predicted knots is roughly flat and unaffected by the number of knot points. The evenly spaced knots have a higher RMSE at low number of knot points. {This occurs because uniform placement does not adapt to the trajectory's jerk activity, leading to poorer fits when knots are sparse.}}
    \label{fig:rmse_nknots}
\end{figure}

\section{Conclusions}
{We demonstrate that a small neural network can predict a jerk-based activity signal and construct adaptive knot distributions from the initial condition. For the same number of knots, the proposed placement reduces the spline approximation error and produces control points closer to the trajectory than uniform placement. The method runs on a Raspberry Pi with a mean runtime of approximately 2~ms, supporting its feasibility for resource-constrained platforms. Future work will focus on integration with real-time trajectory prediction and evaluation on target hardware.}

\section{Acknowledgments}
Copyright 2026 Carnegie Mellon University and Duquesne University

This material is based upon work supported by the Department of War under Air Force Contract No. FA8702-15-D-0002 with Carnegie Mellon University for the operation of the Software Engineering Institute, a federally funded research and development center.  

References herein to any specific entity, product, process, or service by trade name, trademark, manufacturer, or otherwise, does not necessarily constitute or imply its endorsement, recommendation, or favoring by Carnegie Mellon University or its Software Engineering Institute nor of Carnegie Mellon University - Software Engineering Institute by any such named or represented entity.

[DISTRIBUTION STATEMENT A] This material has been approved for public release and unlimited distribution.  Please see Copyright notice for non-US Government use and distribution. DM26-0292
\bibliographystyle{IEEEtran} 
\bibliography{bibliography}
\end{document}